\documentclass[runningheads,envcountsame]{llncs}
\usepackage{graphicx}
\usepackage{mathrsfs}
\usepackage{amsmath}
\usepackage{amssymb}
\usepackage{stmaryrd}
\usepackage[small,bf]{caption}
\usepackage{subfig}
\usepackage{enumitem}
\usepackage{thmtools}
\usepackage{thm-restate}

\makeatletter
\def\moverlay{\mathpalette\mov@rlay}
\def\mov@rlay#1#2{\leavevmode\vtop{%
   \baselineskip\z@skip \lineskiplimit-\maxdimen
   \ialign{\hfil$\m@th#1##$\hfil\cr#2\crcr}}}
\newcommand{\charfusion}[3][\mathord]{
    #1{\ifx#1\mathop\vphantom{#2}\fi
        \mathpalette\mov@rlay{#2\cr#3}
      }
    \ifx#1\mathop\expandafter\displaylimits\fi}
\makeatother

\newcommand{\eps}{\varepsilon}

\makeatletter
\DeclareFontFamily{OMX}{MnSymbolE}{}
\DeclareSymbolFont{MnLargeSymbols}{OMX}{MnSymbolE}{m}{n}
\SetSymbolFont{MnLargeSymbols}{bold}{OMX}{MnSymbolE}{b}{n}
\DeclareFontShape{OMX}{MnSymbolE}{m}{n}{
    <-6>  MnSymbolE5
   <6-7>  MnSymbolE6
   <7-8>  MnSymbolE7
   <8-9>  MnSymbolE8
   <9-10> MnSymbolE9
  <10-12> MnSymbolE10
  <12->   MnSymbolE12
}{}
\DeclareFontShape{OMX}{MnSymbolE}{b}{n}{
    <-6>  MnSymbolE-Bold5
   <6-7>  MnSymbolE-Bold6
   <7-8>  MnSymbolE-Bold7
   <8-9>  MnSymbolE-Bold8
   <9-10> MnSymbolE-Bold9
  <10-12> MnSymbolE-Bold10
  <12->   MnSymbolE-Bold12
}{}

\let\llangle\@undefined
\let\rrangle\@undefined
\DeclareMathDelimiter{\llangle}{\mathopen}%
                     {MnLargeSymbols}{'164}{MnLargeSymbols}{'164}
\DeclareMathDelimiter{\rrangle}{\mathclose}%
                     {MnLargeSymbols}{'171}{MnLargeSymbols}{'171}
\makeatother

\newcommand{\revs}[2]{\mathrm{RevS}(#1,#2)}
\newcommand{\revlc}[1]{\mathrm{RevL}(#1)}
\newcommand{\revl}[2]{\mathrm{RevL}(#1,#2)}
\newcommand{\orevs}[2]{\mathrm{1RevS}(#1,#2)}
\newcommand{\orevl}[2]{\mathrm{1RevL}(#1,#2)}
\newcommand{\supp}{\mathrm{supp}}

\begin{document}
\title{Reversible Weighted Automata over Finite Rings \\ and Monoids with Commuting Idempotents}
 
\titlerunning{Reversible Weighted Automata over Finite Rings}
\author{Peter Kostol\'anyi\orcidID{0000-0002-5474-8781} \and Andrej Ravinger\orcidID{0009-0007-4725-2597}}  
\authorrunning{P. Kostol\'anyi and A. Ravinger}
\institute{Department of Computer Science, Comenius University in Bratislava \\
Mlynsk\'a dolina, 842 48 Bratislava, Slovakia \\
\email{\{kostolanyi,andrej.ravinger\}@fmph.uniba.sk}}
\maketitle              % typeset the header of the contribution
\begin{abstract}
Reversible weighted automata are introduced and~considered in a specific setting where the weights are taken from a nontrivial locally finite commutative ring such as a finite field. 
It~is shown that the supports of series realised by such automata are precisely the rational languages such that the~idempotents in their syntactic monoids commute.
In particular, this is true for reversible weighted automata over the~finite field~$\mathbb{F}_2$, where the realised series can be directly identified with such languages.
A~new automata-theoretic characterisation is thus obtained for the variety of~rational languages corresponding to the pseudovariety of finite monoids~$\mathbf{ECom}$, which also forms the Boolean closure
of the reversible languages in the sense of J.-\'E. Pin. The problem of determining whether a rational series over a locally finite commutative ring can be realised by~a~reversible weighted automaton is decidable as a consequence.
\keywords{Weighted automaton \and Reversible Automaton \and Finite Field \and Variety of Languages \and Decidability}
\end{abstract}

\section{Introduction}

This article embarks upon the study of \emph{reversibility} in~the~context of~\emph{weighted automata}. More precisely, we focus here on~a~very special case of weighted automata over \emph{finite commutative rings},
including in particular the two-element Galois field~$\mathbb{F}_2$. Strong connections to language theory turn out to arise.

The study of the concept of reversibility in computing goes back to the seminal work of~R.~\mbox{Landauer}~\cite{landauer1961a}: according to his fundamental thermodynamical principle,
any loss of information that takes place during a computation necessarily leads to some minimal amount of heat dissipation. This observation led C.~H.~Bennett~\cite{bennett1973a} to consider
logical reversibility of computations, and in particular the \emph{reversible Turing machines}, in which any configuration has at most one preceding and~at~most one successor configuration. 
Similar ideas gradually inspired the~development of~an~entire field of~reversible computing.\vfill\newpage

Starting by D.~Angluin~\cite{angluin1982a} defining -- in the context of language inference -- what is now usually called \emph{bideterministic finite automata}~\cite{pin1992a,tamm2008a,tamm2003a,tamm2004a}, many different variants
of \emph{reversible finite automata} have been introduced and~extensively studied throughout the~literature~\cite{ambainis1998a,golovkins2006a,golovkins2010a,holzer2015a,holzer2018a,holzer2017a,lombardy2002a,pin1992a,radionova2024a,radionova2025a}, while some generalisations of the usual concepts of reversibility in finite automata have been considered as~well~\cite{axelsen2016a,axelsen2017a,garcia2009a,guillon2022a}.
What all these models have in common is that unlike in~the~case of reversible Turing machines, reversibility turns out to be a~real restriction for~finite automata and leads to a loss in their expressivity.
Hence the~study of~the~corresponding classes of languages has been a substantial part of~the~aforementioned efforts.
                                                                                                                   
The notion of a \emph{reversible finite automaton} considered in this article is the one of~J.-\'E. Pin~\cite{pin1992a} -- in other words, a finite automaton is \emph{reversible} if its transition relation is both deterministic
and codeterministic; however, the automaton is allowed to have more than one initial as well as more than one final state. Thus, using the terminology usual in weighted automata theory, we may say that a~reversible finite automaton
is a nondeterministic finite automaton that is finitely sequential~\cite{bala2013a,kostolanyi2022c,kostolanyi2024a,paul2019a,paul2020a,paul2021a} (i.e., ``deterministic'' with possibly more than one initial state) and~its transpose is finitely sequential as well.
This definition is not only very natural, but also leads to a well-behaved class of \emph{reversible languages}, which forms a \emph{positive variety}~\cite{pin1992a}; the~corresponding pseudovariety of~ordered monoids
can be captured by the pseudoinequalities $x^{\omega} y^{\omega} = y^{\omega} x^{\omega}$ and~$x^{\omega} \leq 1$, which also give rise to a natural characterisation of reversible languages in terms of forbidden configurations
in their minimal deterministic automata \cite{pin1992a,klima2018a,klima2020b}. Reversible finite automata with at most one initial and~at~most one final state are precisely the \emph{bideterministic finite automata}~\cite{pin1992a,tamm2008a,tamm2003a,tamm2004a}.
\smallskip

In this article, we lift the notion of reversibility in the sense of~J.-\'E. Pin~\cite{pin1992a} to~the~setting of~\emph{weighted finite automata} -- i.e., finite automata with transitions carrying weights
taken from some algebra such as a semiring, encompassing a~necessary shift from recognising languages towards realisation of~\emph{\mbox{formal} power series} in~several noncommuting variables \cite{berstel2011a,droste2009a,droste2021a,sakarovitch2009a,salomaa1978a}.  
While the~very origins of~weighted automata theory go back already to the seminal article of~\mbox{M.-P.}~\mbox{Sch\"utzenberger}~\cite{schutzenberger1961a}, the field has experienced a wave of renewed interest during the~last decades, and~still represents a~very active area of~study~\mbox{\cite{droste2009a,droste2021a}}.
It is thus relatively surprising that reversible weighted automata have attracted almost no attention so far, the~only exception being the~recent research on~bideterministic weighted automata \cite{kostolanyi2022b,kostolanyi2023c}. 

The study of reversible weighted automata would not only be natural in~view of~the~previous research on~reversibility in~automata theory, but it can also be motivated by the~study of~certain \emph{decision problems} for weighted automata.
The~\mbox{\emph{determinisability}} problem, in which one asks about the~existence of~a~deterministic -- or sequential~\cite{lombardy2006a} -- equivalent of~a~given weighted automaton, attracted significant attention in recent years \cite{bell2023a,benalioua2024a,jecker2024a,kostolanyi2022a}.
While the problem was proved to be decidable in~important settings such as over fields~\cite{bell2023a} and~some complexity bounds have been obtained as well~\cite{benalioua2024a,jecker2024a}, one still has no efficient algorithms for deciding the~determinisability problem for sufficiently general classes of weighted automata.
On the other hand, it has been shown that one can do better when deciding the~existence of~a~\emph{bideterministic} equivalent for~a~given weighted automaton~\cite{kostolanyi2022b,kostolanyi2023c}. This indicates that it might make sense to also look at some other
restrictions of weighted automata related to determinism, and in particular at the \emph{reversible weighted automata}, which represent both a~\mbox{natural} generalisation of~bideterministic weighted automata, as~well~as a~class of finitely sequential automata that is in general incomparable with the~deterministic weighted automata when it comes to expressivity.

The aim of this article is to initiate a systematic research on reversible weighted automata, while we mostly focus here on~a~very special case, in which the weights are from a \emph{locally finite commutative ring}.
This includes the~case of~weighted automata and formal power series over \emph{finite fields}, and~in~particular over the~\emph{\mbox{two-element} field $\mathbb{F}_2$}. The latter setting gives rise to a~natural
way of~describing languages, as any rational series over $\mathbb{F}_2$ is at~the~same time a~characteristic series of~some rational language -- in fact, the series over $\mathbb{F}_2$ correspond to what has been studied as \emph{formal languages over $\mathrm{GF}(2)$} by~E.~\mbox{Bakinova}~et~al.~\cite{bakinova2022a,okhotin2024a};
the idea of describing languages using weighted automata over fields also appears in~the~study of~\emph{image-binary automata} of S.~Kiefer and~C.~\mbox{Widdershoven}~\mbox{\cite{kiefer2021a,kiefer2024a}}. 
In~a~similar spirit, a~weighted automaton over a~\emph{finite} or \emph{locally finite} (semi)ring gives rise to a~\emph{rational} language by~taking the~support of its behaviour, i.e., the~language of~all words with 
a~nonzero coefficient in~the~realised series. 
\smallskip
  
We show that regardless of a \emph{nontrivial locally finite commutative ring} $R$ considered, the languages described by the~\emph{reversible weighted automata} over $R$ in this way always form the same class, 
namely the variety of languages corresponding to the pseudovariety $\mathbf{ECom}$ of all finite monoids with commuting idempotents. This is at~the~same time 
the Boolean closure of the positive variety of all reversible languages~\cite{pin1992a}, i.e., the variety generated by such languages.

Note that the result applies in particular to reversible weighted automata over $\mathbb{F}_2$ -- the languages described by such automata are precisely the languages from the aforementioned variety.
This means that interpreting a reversible finite automaton as a reversible weighted automaton over $\mathbb{F}_2$ leads to an increase in~the~expressive power of~the~model
and to better closure properties of~the~corresponding class of languages. At the same time, a~new automata-theoretic characterisation of the variety of languages corresponding to $\mathbf{ECom}$ is obtained.      

We also characterise the \emph{series} realised by the reversible weighted automata over nontrivial \emph{finite} commutative rings. This characterisation implies, together with effective decidability of~membership of a~monoid
to $\mathbf{ECom}$, the~existence of~an~algorithm for~deciding whether a weighted automaton over an~effective \emph{locally finite} commutative ring~$R$ admits a~reversible equivalent over $R$ or not.               

\section{Preliminaries} 
                                                  
We denote by $\mathbb{N}$ and $\mathbb{Q}$ the sets of all \emph{nonnegative} integers and rational numbers, and by $\mathbb{B}$ the Boolean domain $\mathbb{B} = \{0,1\}$.
We write $[n] = \{1,\ldots,n\}$ for all $n \in \mathbb{N}$. Alphabets are assumed to be finite and nonempty, $\eps$ denotes the empty word.\goodbreak

A \emph{semiring} is a quintuple $(S,+,\cdot,0,1)$, or simply $S$, such that $(S,+,0)$ is a commutative monoid, $(S,\cdot,1)$ is a monoid,
$\cdot$ distributes over $+$ from both sides, and $0$ is a zero in $(S,\cdot,1)$. A \emph{subsemiring} of $(S,+,\cdot,0,1)$ is a semiring $(T,+_T,\cdot_T,0,1)$, where $T \subseteq S$ and $+_T,\cdot_T$ are the restrictions of $+$ and~$\cdot$~to~$T$. 
The~\emph{subsemiring of $S$ generated by $G \subseteq S$} is the smallest subsemiring $\langle G\rangle$ of~$S$ such that $G \subseteq \langle G\rangle$. A semiring $(S,+,\cdot,0,1)$ is \emph{commutative} when $\cdot$ is, \emph{finite} when~$S$~is,
\emph{finitely generated} when $S = \langle G\rangle$ for some finite $G \subseteq S$, and \emph{locally finite} when all finitely generated subsemirings of $S$ are finite.
A semiring $(S,+,\cdot,0,1)$ is \emph{nontrivial} if~$S$ contains at least two elements, that is, if \mbox{$0 \neq 1$}.
An~important example of a semiring is the \emph{Boolean semiring} $(\mathbb{B},\lor,\land,0,1)$.

A \emph{ring} (with unity) is a semiring $(R,+,\cdot,0,1)$ such that $(R,+,0)$ is an abelian group. 
A \emph{field} is a nontrivial commutative ring $(\mathbb{F},+,\cdot,0,1)$ such that $(\mathbb{F} \setminus \{0\},\cdot,1)$ is an abelian group. 
The~finite field with two elements is denoted by $\mathbb{F}_2$.

The reader can consult \cite{berstel2011a,droste2009a,droste2021a,sakarovitch2009a,salomaa1978a} for the basics on formal power series in several noncommuting variables and weighted automata.
A~\emph{formal power series} over a semiring $S$ and alphabet $\Sigma$ is a mapping $r\colon\Sigma^* \to S$; the value of~$r$ upon $w \in \Sigma^*$ is denoted by~$(r,w)$ and~called
the \emph{coefficient} of $r$ at $w$, while one writes
$r = \sum_{w \in \Sigma^*} (r,w)\,w$.
The~set of~all series over $S$ and $\Sigma$ is denoted by~$S\llangle\Sigma^*\rrangle$. Given $r,s \in S\llangle\Sigma^*\rrangle$, we define $r + s$ by~$(r + s, w) = (r,w) + (s,w)$ and $r \cdot s$ by~$(r\cdot s,w) = \sum_{u,v \in \Sigma^*, uv = w} (r,u)(s,v)$ for all $w \in \Sigma^*$.
Each $a \in S$ is identified with $r_a \in S\llangle\Sigma^*\rrangle$ such that $(r_a,\eps) = a$ and~$(r_a,w) = 0$ for~all $w \in \Sigma^+$; the~left or right multiplication by~$a$ thus corresponds to left or right \emph{scalar multiplication}.
Similarly, each $w \in \Sigma^*$ is identified with $r_w \in S\llangle\Sigma^*\rrangle$ such that $(r_w,w) = 1$ and~$(r_w,x) = 0$ for~all $x \in \Sigma^* \setminus \{w\}$.

The \emph{support} of $r \in S\llangle\Sigma^*\rrangle$ is a language $\supp(r) = \{w \in \Sigma^* \mid (r,w) \neq 0\}$, and~the~\emph{characteristic series} $\underline{L} \in S\llangle\Sigma^*\rrangle$ of~a~language $L \subseteq \Sigma^*$ over a~nontrivial semiring $S$ is
given by $(\underline{L},w) = 1$ for $w \in L$ and $(\underline{L},w) = 0$ for $w \in \Sigma^* \setminus L$.

A \emph{nondeterministic finite automaton} over an alphabet $\Sigma$ is a quadruple \mbox{$\mathcal{A} = (Q,\rightarrow,I,F)$}, where $Q$ is a~finite set of~states, ${\rightarrow} \subseteq Q \times \Sigma \times Q$ is a transition relation, and $I,F \subseteq Q$ are
the sets of initial and final states, respectively. Given $p,q \in Q$ and $w \in \Sigma^*$ such that $w = a_1 \ldots a_n$ for some $n \in \mathbb{N}$ and $a_1,\ldots,a_n \in \Sigma$, we write $p \stackrel{w}{\rightarrow} q$ if there are states $p_0,\ldots,p_n \in Q$ such that $p_0 = p$, $p_n = q$, and~$(p_{k - 1},a_k,p_k) \in {\rightarrow}$ for $k = 1,\ldots,n$. 
The \emph{language} described by $\mathcal{A}$ is given by $\|\mathcal{A}\| = \{w \in \Sigma^* \mid \exists p \in I~\exists q \in F: p \stackrel{w}{\rightarrow} q\}$.
The automaton $\mathcal{A}$ is \emph{deterministic} if $\lvert I\rvert = 1$ and $q = q'$ whenever $p \stackrel{a}{\rightarrow} q$ and $p \stackrel{a}{\rightarrow} q'$ for some $p,q,q' \in Q$ and $a \in \Sigma$.

A \emph{deterministic finite automaton} with a complete transition function over $\Sigma$ can also be written as $\mathcal{A} = (Q,\cdot,i,F)$, where $\cdot\colon Q \times \Sigma^* \to Q$ is a~right action of~$\Sigma^*$ on~$Q$, where $i \in Q$ is the~initial state, and~$F \subseteq Q$ is a~set of~final states.
The~\emph{language} recognised by~$\mathcal{A}$ is given by~$\|\mathcal{A}\| = \{w \in \Sigma^* \mid i \cdot w \in F\}$.\goodbreak 
 
A \emph{weighted automaton} $\mathcal{A} = (Q,\sigma,\iota,\tau)$ over a semiring $S$ and alphabet $\Sigma$ is given by a finite set of states $Q$, a transition weighting function $\sigma\colon Q \times \Sigma \times Q \to S$,
and functions $\iota\colon Q \to S$, $\tau\colon Q \to S$ assigning initial and final weights to states. A \emph{run} of $\mathcal{A}$ is a word $\gamma = q_0 a_1 q_1 a_2 q_2 \ldots q_{t-1} a_t q_t$ with $t \in \mathbb{N}$, $q_0,\ldots,q_t \in Q$, and $a_1,\ldots,a_t \in \Sigma$
such that $\sigma(q_{k - 1},a_k,q_k) \neq 0$ for $k = 1,\ldots,t$; we say that $\gamma$ is a run on $w = a_1 \ldots a_t$ leading from $q_0$ to $q_t$. Given any such run $\gamma$, we set $\sigma(\gamma) = \sigma(q_0,a_1,q_1) \ldots \sigma(q_{t-1},a_t,q_t)$ and $\overline{\sigma}(\gamma) = \iota(q_0)\sigma(\gamma)\tau(q_t)$.
Let $\mathcal{R}(\mathcal{A},w)$ be the set of all runs of $\mathcal{A}$ on $w \in \Sigma^*$. The \emph{series $\|\mathcal{A}\| \in S\llangle\Sigma^*\rrangle$ realised by $\mathcal{A}$} is then given by $(\|\mathcal{A}\|,w) = \sum_{\gamma \in \mathcal{R}(\mathcal{A},w)} \overline{\sigma}(\gamma)$ for all $w \in \Sigma^*$. 
A~series $r \in S\llangle\Sigma^*\rrangle$ is \emph{rational} over $S$ if $r = \|\mathcal{A}\|$ for some weighted automaton $\mathcal{A}$ over $S$.  

Every $\mathcal{A} = (Q,\sigma,\iota,\tau)$ over $S$ and $\Sigma$ such that $Q = [n]$ for some $n \in \mathbb{N}$ (which can always be assumed) determines a \emph{linear representation}
$\mathcal{P}_{\mathcal{A}} = (n,\mathbf{i},\mu,\mathbf{f})$, where $\mathbf{i} = (\iota(1),\ldots,\iota(n))$ is a row vector of initial weights, $\mu$~is~a~homomorphism from $\Sigma^*$ to the monoid $S^{n \times n}$ of
all $n \times n$ matrices over $S$ with matrix multiplication such that $\mu(a) = (\sigma(i,a,j))_{n \times n}$ for~all $a \in \Sigma$, and $\mathbf{f} = (\tau(1),\ldots,\tau(n))^T$ is a~column vector of final weights.
One then has $(\|\mathcal{A}\|,w) = \mathbf{i} \mu(w) \mathbf{f}$ for all $w \in \Sigma^*$.

Let $\mathcal{A}_k = (Q_k,\sigma_k,\iota_k,\tau_k)$ be weighted automata over $S$ and $\Sigma$ for $k = 1,\ldots,n$, where $n \in \mathbb{N}$. The \emph{disjoint union} of $\mathcal{A}_1,\ldots,\mathcal{A}_n$ is an automaton
$\mathcal{A} = (Q,\sigma,\iota,\tau)$, where $Q = \bigcup_{k = 1}^n \left(Q_k \times \{k\}\right)$, $\sigma((p,k),a,(q,k)) = \sigma_k(p,a,q)$, $\iota(p,k) = \iota_k(p)$, and $\tau(q,k) = \tau_k(q)$ for $k = 1,\ldots,n$ and all $p,q \in Q_k$ and $a \in \Sigma$;
moreover, $\sigma((p,k),a,(q,\ell)) = 0$ for all $k,\ell \in [n]$ such that $k \neq \ell$, and all $p \in Q_k$, $q \in Q_{\ell}$, and $a \in \Sigma$. Clearly $\|\mathcal{A}\| = \|\mathcal{A}_1\| + \ldots + \|\mathcal{A}_n\|$.  

We also use some elementary concepts from \emph{algebraic language theory}; we refer the reader to \cite{pin1997a,straubing2021a} for the basic theory of \emph{varieties of languages}, and~to~\cite{almeida1994a}
for the theory of \emph{pseudovarieties of finite monoids}.

\section{Reversible Weighted Automata}

We now define the \emph{reversible weighted automata} over a semiring $S$ by generalising the notion of~reversible finite automata as~understood by~J.-\'E. Pin~\cite{pin1992a}.

\begin{definition}
Let $S$ be a semiring and $\Sigma$ an alphabet. A weighted automaton $\mathcal{A} = (Q,\sigma,\iota,\tau)$ over $S$ and~$\Sigma$ is \emph{reversible} if the following two conditions are satisfied for all $p,p',q,q' \in Q$ and $a \in \Sigma$:
\begin{enumerate}[label=$(\roman*)$]
\item{If $\sigma(p,a,q)$ and $\sigma(p,a,q')$ are both nonzero, then $q = q'$; 
}
\item{If $\sigma(p,a,q)$ and $\sigma(p',a,q)$ are both nonzero, then $p = p'$.
}
\end{enumerate}
\end{definition}
 
The condition $(i)$ says that the transitions of the automaton $\mathcal{A}$ have to be deterministic; as there still can be more than one initial state, this is equivalent to saying that $\mathcal{A}$ is \emph{finitely sequential}~\cite{bala2013a,kostolanyi2022c,kostolanyi2024a,paul2019a,paul2020a,paul2021a}.
The condition $(ii)$ says that the transpose of $\mathcal{A}$ has the same property. If in addition there is at most one state $p \in Q$ with $\iota(p) \neq 0$ and at most one state $q \in Q$ with $\tau(q) \neq 0$, the~weighted automaton $\mathcal{A}$ is called \emph{bideterministic}~\cite{kostolanyi2022b,kostolanyi2023c}.

We call a series over a semiring $S$ and over an alphabet $\Sigma$ \emph{reversible} over~$S$ if~it~is realised by some reversible weighted automaton over $S$ and $\Sigma$. 
The~set of~all reversible series over $S$ and~$\Sigma$ is denoted by $\revs{S}{\Sigma}$. Moreover, by~a~\emph{\mbox{reversible} language over $S$ and $\Sigma$}, we understand a support of~some reversible series over~$S$~and~$\Sigma$;
we write
$\revl{S}{\Sigma} = \{\supp(r) \mid r \in \revs{S}{\Sigma}\}$ for the~set of~all such languages,
and $\revlc{S}$ for the class of all reversible languages over~$S$ and any alphabet. 
Observe that the reversible languages over the~Boolean semi\-ring $\mathbb{B}$ are precisely the~usual reversible languages in~the~sense of~J.-\'E.~Pin~\cite{pin1992a}.
We thus denote the positive variety of all such languages by $\revlc{\mathbb{B}}$.\goodbreak      
            
\begin{example}
We mostly consider a very special class of reversible weighted automata over \emph{finite} -- or slightly more generally, \emph{locally finite} -- \emph{commutative rings} in this article.
What can be readily observed right now is that given any nontrivial locally finite ring $R$, the class $\revlc{R}$ of all reversible languages over $R$ contains at least some languages that are not reversible in the usual sense, i.e., which do not belong to $\revlc{\mathbb{B}}$. 
Indeed, let $R$ be of characteristic $n \geq 2$. Then the reversible weighted automaton $\mathcal{A}$ in~Fig.~\ref{fig:1} clearly realises the~series
\begin{displaymath}
\|\mathcal{A}\| = \sum_{t \in \mathbb{N} \setminus \{0\}} a^t.  
\end{displaymath}

As a consequence, $\supp(\|\mathcal{A}\|) = a^+$ is a reversible language over $R$, so~that $a^+ \in \revl{R}{\{a\}}$. On~the~other hand, $a^+ \not\in \revl{\mathbb{B}}{\{a\}}$, as it does not satisfy the~characterisation of reversible languages from~\cite{pin1992a}; more precisely, the~ordered syntactic monoid of $a^+$ does not satisfy the pseudoinequality $x^{\omega} \leq 1$.

\begin{figure}[h]
\vspace{-14pt}
\begin{center}
\includegraphics{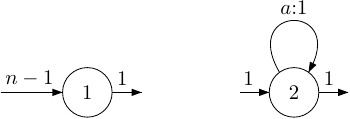}
\vspace{-12pt}
\end{center}
\caption{\label{fig:1}A reversible weighted automaton $\mathcal{A}$ over a nontrivial locally finite ring $R$ of~characteristic $n \geq 2$ and over a~unary alphabet $\Sigma = \{a\}$.}
\vspace{-20pt}
\end{figure}
\end{example}  

Given the observation from the previous example, it seems to be worthwhile to take a closer look at the properties of the classes of reversible series and~languages introduced above.
At least a few observations can be made at~the~abstract level of semirings. Given any semiring $S$ and alphabet $\Sigma$, we denote by $\orevs{S}{\Sigma}$ the set of all series realised by reversible weighted automata $\mathcal{A} = (Q,\sigma,\iota,\tau)$ over $S$ and $\Sigma$
with \emph{precisely one initial state}, i.e., a~state~$q$ \mbox{satisfying} $\iota(q) \neq 0$.\footnote{Note that such automata form a weighted generalisation of reversible deterministic finite automata in the sense of M. Holzer, S. Jakobi, and M. Kutrib \cite{holzer2015a,holzer2018a}.} We also write
$\orevl{S}{\Sigma} = \{\supp(r) \mid r \in \orevs{S}{\Sigma}\}$
for~the~corresponding set of supports. We then have the following obvious characterisation of~the~reversible series over $S$.

\begin{proposition}
\label{prop:revs-char}
Let $S$ be a semiring and $\Sigma$ an alphabet. Then $\revs{S}{\Sigma}$ consists of precisely all finite sums of series from $\orevs{S}{\Sigma}$.
\end{proposition}
\begin{proof}
Let $\mathcal{A} = (Q,\sigma,\iota,\tau)$ be a reversible weighted automaton over $S$ and $\Sigma$. The~automaton $\mathcal{A}_q = (Q,\sigma,\iota_q,\tau)$,
with $\iota_q(q) = \iota(q)$ and $\iota_q(p) = 0$ for all \mbox{$p \in Q \setminus \{q\}$}, is then reversible with precisely one initial state for each $q \in Q$ such that $\iota(q) \neq 0$; as a result, $\|\mathcal{A}_q\| \in \orevs{S}{\Sigma}$.
Thus
\begin{displaymath}
\|\mathcal{A}\| = \sum_{\substack{q \in Q \\ \iota(q) \neq 0}} \|\mathcal{A}_q\|
\end{displaymath}
is a finite sum of series from $\orevs{S}{\Sigma}$. Conversely, given $k \in \mathbb{N}$ and~series $r_1,\ldots,r_k \in \orevs{S}{\Sigma}$, each of the series $r_j$ for $j = 1,\ldots,k$ is realised by some reversible weighted automaton
$\mathcal{A}_j$ over $S$ and $\Sigma$. The disjoint union of~these automata is then clearly reversible as well, hence $r_1 + \ldots + r_k \in \revs{S}{\Sigma}$.\qed   
\end{proof}\goodbreak

\begin{proposition}
\label{prop:closure-sum}
The set $\revs{S}{\Sigma}$ is closed under addition and under -- both left and right -- scalar multiplication for every semiring~$S$ and alphabet $\Sigma$.
\end{proposition}
\begin{proof}
Closure under addition follows directly by Proposition \ref{prop:revs-char}. For left or right scalar multiplication by $\alpha \in S$, it is clearly sufficient to multiply all initial or final weights of~a~reversible weighted automaton by $\alpha$.\qed
\end{proof}

Any finite automaton can be turned into a weighted automaton over a~nontrivial semiring $S$ by assigning the weight $1$ to all its transitions, initial states, and final states.
It is standard that for \emph{deterministic} finite automata recognising some language $L$, the resulting automaton realises the characteristic series~$\underline{L}$ of $L$ over $S$.  
The following proposition records this observation specifically for~the~reversible automata with one initial state. 

\begin{proposition}
\label{prop:char-ser}
Let $S$ be a nontrivial semiring, $L \in \orevl{\mathbb{B}}{\Sigma}$, and $\underline{L}$ be the~characteristic series of $L$ over $S$. Then $\underline{L} \in \orevs{S}{\Sigma}$. As a consequence, $L \in \orevl{S}{\Sigma}$ and $\orevl{\mathbb{B}}{\Sigma} \subseteq \orevl{S}{\Sigma}$. 
\end{proposition}
\begin{proof}
As $L \in \orevl{\mathbb{B}}{\Sigma}$, the language $L$ is surely recognised by some reversible finite automaton $\mathcal{A} = (Q,\rightarrow,I,F)$ with precisely one initial state \mbox{$i \in I$}.
Let $\mathcal{A}' = (Q,\sigma,\iota,\tau)$ be a~weighted automaton such that for all $p,q \in Q$ and~\mbox{$a \in \Sigma$}, one has $\sigma(p,a,q) = 1$ if $p \stackrel{a}{\rightarrow} q$ and~$\sigma(p,a,q) = 0$ otherwise,
$\iota(p) = 1$ if $p \in I$ and~$\iota(p) = 0$ otherwise, and~$\tau(q) = 1$ if $q \in F$ and~$\tau(q) = 0$ otherwise. Since $\mathcal{A}$ is deterministic, surely $\|\mathcal{A}'\| = \underline{L}$, while the automaton
$\mathcal{A}'$ is clearly reversible with precisely one initial state $i$. As a result, it follows that $\underline{L} \in \orevs{S}{\Sigma}$ and~$L = \supp(\underline{L}) \in \orevl{S}{\Sigma}$.\qed
\end{proof}

\section{The Case of $\mathbb{F}_2$}

We now explore in detail the \emph{reversible weighted automata} over the \emph{two-element field $\mathbb{F}_2$}. The field $\mathbb{F}_2$ is arguably the most important finite ring from 
the viewpoint of weighted automata theory, and weighted automata over $\mathbb{F}_2$ can also be seen as natural devices for describing languages. Indeed, any formal power series over $\mathbb{F}_2$
is a characteristic series of its support, and both can be identified as a result. The rational series over $\mathbb{F}_2$ can thus be seen as rational languages, and the realisation of~a~rational series by~a~weighted automaton over $\mathbb{F}_2$ corresponds to recognising a~rational language in~``parity mode'', where a word $w$ gets accepted if and only if there is an odd number of successful runs on~$w$ in~the~automaton. 
In fact, the theory of formal power series over $\mathbb{F}_2$ viewed as languages has largely been developed under the name \emph{formal languages over $\mathrm{GF}(2)$}~\cite{bakinova2022a,okhotin2024a},
and weighted automata over $\mathbb{F}_2$ were also studied under the name \emph{symmetric difference automata} \cite{merwe2012a,zijl2004a}.\goodbreak 
                                                                                                                                                 
Similar observations can be done in particular for reversible weighted automata over $\mathbb{F}_2$: the series from $\revs{\mathbb{F}_2}{\Sigma}$ might be, for any alphabet $\Sigma$, identified with languages from $\revl{\mathbb{F}_2}{\Sigma}$.
It thus seems to be natural to ask about a characterisation of the language class $\revlc{\mathbb{F}_2}$ and about its relation to the usual class of reversible languages $\revlc{\mathbb{B}}$. 
We prove in this section that $\revlc{\mathbb{F}_2}$ is in fact the \emph{variety} generated by the positive variety $\revlc{\mathbb{B}}$ -- or~equivalently, the \emph{Boolean closure} of $\revlc{\mathbb{B}}$.
\smallskip

Let us start by recording an easy observation that reversible weighted automata over $\mathbb{F}_2$ with \emph{one initial state} realise precisely the~characteristic series
of~languages recognised by reversible finite automata with one initial state.

\begin{proposition}
\label{prop:chs2}
Let $\Sigma$ be an alphabet and $\mathcal{A}$ a reversible weighted automaton over $\mathbb{F}_2$ and $\Sigma$ with one initial state. Then $\|\mathcal{A}\| = \underline{L}$ for some $L \in \orevl{\mathbb{B}}{\Sigma}$ over $\mathbb{F}_2$.
As a consequence, $\orevl{\mathbb{F}_2}{\Sigma} = \orevl{\mathbb{B}}{\Sigma}$.
\end{proposition}
\begin{proof}
Let $\mathcal{A} = (Q,\sigma,\iota,\tau)$ be a reversible weighted automaton with one initial state over $\mathbb{F}_2$ and $\Sigma$.
As $\mathcal{A}$ is deterministic, the reversible finite automaton $\mathcal{A}' = (Q,\rightarrow,I,F)$ such that one has $p \stackrel{a}{\rightarrow} q$ for $p,q \in Q$ and $a \in \Sigma$ if~and~only~if $\sigma(p,a,q) = 1$, while $I = \{p \in Q \mid \iota(p) = 1\}$ and $F = \{q \in Q \mid \tau(q) = 1\}$, clearly recognises $L = \supp(\|\mathcal{A}\|)$. Hence $L \in \orevl{\mathbb{B}}{\Sigma}$ and $\|\mathcal{A}\| = \underline{L}$. The equality $\orevl{\mathbb{F}_2}{\Sigma} = \orevl{\mathbb{B}}{\Sigma}$ then follows by~Proposition~\ref{prop:char-ser}.\qed
\end{proof}

Recall from \cite{axelsen2017a} that the class of languages recognised by reversible finite automata with precisely one initial state is closed under intersection. 
Indeed, if \mbox{$\mathcal{A}_1 = (Q_1,\rightarrow_1,\{i_1\},F_1)$} and $\mathcal{A}_2 = (Q_2,\rightarrow_2,\{i_2\},F_2)$ with $i_1 \in Q_1$ and $i_2 \in Q_2$ are
reversible finite automata over an alphabet $\Sigma$, then the finite automaton $\mathcal{A} = (Q_1 \times Q_2, \rightarrow, \{(i_1,i_2)\}, F_1 \times F_2)$, with $\rightarrow$ given for all $p_1,q_1 \in Q_1$, $p_2,q_2 \in Q_2$, and $a \in \Sigma$ by
$(p_1,p_2) \stackrel{a}{\rightarrow} (q_1,q_2)$ if and only if $p_1 \stackrel{a}{\rightarrow}_1 q_1$ and $p_2 \stackrel{a}{\rightarrow}_2 q_2$, is clearly reversible as well.
Let us record this observation for later reference.  

\begin{proposition}[H. B. Axelsen, M. Holzer, and M. Kutrib \cite{axelsen2017a}]     
\label{prop:ahk}
The set $\orevl{\mathbb{B}}{\Sigma}$ is closed under intersection for every alphabet $\Sigma$.
\end{proposition}
     
We are now prepared to explore some basic closure properties of~the~class $\revlc{\mathbb{F}_2}$ of reversible languages over $\mathbb{F}_2$. 
First of all, Proposition \ref{prop:closure-sum} directly implies closure of $\revlc{\mathbb{F}_2}$ under symmetric difference. Somewhat more interesting is the closure
of this class under all Boolean operations, which we now establish.

\begin{proposition}
\label{prop:clos}
Let $\Sigma$ be an alphabet. The set $\revl{\mathbb{F}_2}{\Sigma}$ is then closed under the Boolean operations.
\end{proposition}
\begin{proof}
Let $L \in \revl{\mathbb{F}_2}{\Sigma}$. Then $L$ is a support of some series in $\revs{\mathbb{F}_2}{\Sigma}$, which has to be the~characteristic series~$\underline{L}$ of~$L$ over~$\mathbb{F}_2$. 
Thus $\underline{L} \in \revs{\mathbb{F}_2}{\Sigma}$. Moreover, it is easy to see that the characteristic series $\underline{\Sigma^*}$ of $\Sigma^*$ over $\mathbb{F}_2$ is realised by a reversible weighted automaton over $\mathbb{F}_2$ and $\Sigma$
with one state, so~that $\underline{\Sigma^*} \in \revs{\mathbb{F}_2}{\Sigma}$. It thus follows by Proposition \ref{prop:closure-sum} that the series 
\begin{displaymath}
\underline{L} + \underline{\Sigma^*} = \sum_{w \in \Sigma^*}\left((\underline{L},w) + 1\right)\,w
\end{displaymath}   
belongs to $\revs{\mathbb{F}_2}{\Sigma}$. As a result, 
\begin{displaymath}
\supp(\underline{L} + \underline{\Sigma^*}) = \{w \in \Sigma^* \mid (\underline{L},w) + 1 \neq 0\} = \{w \in \Sigma^* \mid (\underline{L},w) = 0\} = \Sigma^* \setminus L 
\end{displaymath}
belongs to $\revl{\mathbb{F}_2}{\Sigma}$, and $\revl{\mathbb{F}_2}{\Sigma}$ is closed under complementation.\goodbreak

Next, let $L,K \in \revl{\mathbb{F}_2}{\Sigma}$. Again, this means that the characteristic series $\underline{L},\underline{K}$ of these two languages over $\mathbb{F}_2$ are in~$\revs{\mathbb{F}_2}{\Sigma}$.
By~Proposition~\ref{prop:revs-char} and~Proposition~\ref{prop:chs2}, there have to exist some numbers $m,n \in \mathbb{N}$ and languages \mbox{$L_1,\ldots,L_m,K_1,\ldots,K_n \in \orevl{\mathbb{B}}{\Sigma}$} such that over $\mathbb{F}_2$,
\begin{displaymath}
\underline{L} = \underline{L_1} + \ldots + \underline{L_m} \qquad \text{ and } \qquad \underline{K} = \underline{K_1} + \ldots + \underline{K_n}. 
\end{displaymath}  
Now, for every $w \in \Sigma^*$, one has
\begin{align*}
(\underline{L \cap K}, w) & = (\underline{L_1} + \ldots + \underline{L_m}, w) \cdot (\underline{K_1} + \ldots + \underline{K_n}, w) = \\
& = \left((\underline{L_1},w) + \ldots + (\underline{L_m}, w)\right) \cdot \left((\underline{K_1},w) + \ldots + (\underline{K_n}, w)\right) = \\
& = \sum_{i = 1}^m \sum_{j = 1}^n (\underline{L_i}, w)(\underline{K_j},w) = \sum_{i = 1}^m \sum_{j = 1}^n (\underline{L_i \cap K_j}, w)    
\end{align*}
over $\mathbb{F}_2$, so that
\begin{displaymath}
\underline{L \cap K} = \sum_{i = 1}^m \sum_{j = 1}^n \underline{L_i \cap K_j}. 
\end{displaymath} 
As $L_i \cap K_j \in \orevl{\mathbb{B}}{\Sigma}$ by Proposition \ref{prop:ahk} for $i = 1,\ldots,m$ and $j = 1,\ldots,n$, it follows by Proposition \ref{prop:char-ser} that $\underline{L_i \cap K_j} \in \orevs{\mathbb{F}_2}{\Sigma}$,
and as a consequence, $\underline{L \cap K} \in \revs{\mathbb{F}_2}{\Sigma}$ by Proposition \ref{prop:revs-char}. As a result, $\supp(\underline{L \cap K}) = L \cap K$ has to be in $\revl{\mathbb{F}_2}{\Sigma}$, which is thus closed under intersection.

Finally, the closure of $\revl{\mathbb{F}_2}{\Sigma}$ under union follows by its closure under intersection and under complementation.\qed
\end{proof}

As a first step towards a characterisation of the class $\revlc{\mathbb{F}_2}$ of reversible languages over $\mathbb{F}_2$, let us prove that this class contains
all reversible languages in the usual sense -- that is, all languages from $\revlc{\mathbb{B}}$.

\begin{lemma}
\label{lem:ma}
Let $\Sigma$ be an alphabet. Then $\revl{\mathbb{B}}{\Sigma} \subseteq \revl{\mathbb{F}_2}{\Sigma}$.
\end{lemma}
\begin{proof}
Let $L \in \revl{\mathbb{B}}{\Sigma}$. Then $L = L_1 \cup \ldots \cup L_n$ for some $n \in \mathbb{N}$ and languages $L_1,\ldots,L_n \in \orevl{\mathbb{B}}{\Sigma}$.
The series
\begin{displaymath}
r = \sum_{\emptyset \subsetneq X \subseteq [n]} \underline{\bigcap_{i \in X} L_i}
\end{displaymath}
over $\mathbb{F}_2$ then belongs to $\revs{\mathbb{F}_2}{\Sigma}$ by~virtue of~Proposition~\ref{prop:ahk}, Proposition~\ref{prop:char-ser}, and~Proposition~\ref{prop:closure-sum}.  
This means that in order to prove $L \in \revl{\mathbb{F}_2}{\Sigma}$, it is sufficient to show that $\underline{L} = r$ over $\mathbb{F}_2$. 

Given $w \in \Sigma^*$, let $X_w = \{i \in [n] \mid w \in L_i\}$. Then for every nonempty $X \subseteq [n]$, the word $w$ belongs to $\bigcap_{i \in X} L_i$ if and only if $X \subseteq X_w$.
The total number of such nonempty sets $X$ is thus $2^{\lvert X_w\rvert} - 1$, which is odd if $w \in L$ and zero otherwise. As~a~result,
\begin{align*}
(r,w) & = \left(\sum_{\emptyset \subsetneq X \subseteq [n]} \underline{\bigcap_{i \in X} L_i}, w\right) = \sum_{\emptyset \subsetneq X \subseteq [n]} \left(\underline{\bigcap_{i \in X} L_i}, w\right) = \\
& = \sum_{\substack{\emptyset \subsetneq X \subseteq X_w}} 1 = 2^{\lvert X_w\rvert} - 1 = (\underline{L},w)
\end{align*}                                                                                        
over $\mathbb{F}_2$, which establishes the desired equality $\underline{L} = r$.\qed
\end{proof} 

We are now finally ready to prove the characterisation of the class $\revlc{\mathbb{F}_2}$, which happens to be the~Boolean closure of the positive variety of reversible languages $\revlc{\mathbb{B}}$.

\begin{theorem}
\label{th:f2hlavna}
The class $\revlc{\mathbb{F}_2}$ is the Boolean closure of $\revlc{\mathbb{B}}$.
\end{theorem}
\begin{proof}
The Boolean closure $\mathbf{B}(\revlc{\mathbb{B}})$ of $\revlc{\mathbb{B}}$ is included in $\revlc{\mathbb{F}_2}$ by~Lemma~\ref{lem:ma} and~Proposition~\ref{prop:clos}. 
For the converse, let $L \in \revl{\mathbb{F}_2}{\Sigma}$ for~some alphabet $\Sigma$. Proposition~\ref{prop:revs-char} gives us $n \in \mathbb{N}$ 
and~$r_1,\ldots,r_n \in \orevs{\mathbb{F}_2}{\Sigma}$ such that $L = \supp(r_1 + \ldots + r_n)$. We prove that $L$ is in $\mathbf{B}(\revlc{\mathbb{B}})$ by induction on $n$. If~$n = 0$, then $L = \supp(0) = \emptyset$ is in $\mathbf{B}(\revlc{\mathbb{B}})$.
Now, for any $k \in \mathbb{N}$, $n = k + 1$, and $L_1 := \supp(r_1 + \ldots + r_k)$ in $\mathbf{B}(\revlc{\mathbb{B}})$, we see that $L_2 := \supp(r_{k+1})$ is in~$\mathbf{B}(\revlc{\mathbb{B}})$, as it in fact belongs to $\orevl{\mathbb{B}}{\Sigma}$ by Proposition \ref{prop:chs2}.
Thus 
\begin{displaymath}
L = \supp(\left(r_1 + \ldots + r_k\right) + r_{k + 1}) = (L_1 \cup L_2) \cap \left(\Sigma^* \setminus (L_1 \cap L_2)\right)
\end{displaymath} 
belongs to $\mathbf{B}(\revlc{\mathbb{B}})$ as well.\qed
\end{proof}

The class $\revlc{\mathbb{F}_2}$ of all reversible languages over the field $\mathbb{F}_2$ thus forms a~\emph{variety of~languages} given by the~Boolean closure of the positive variety $\revlc{\mathbb{B}}$ of~the~usual reversible languages
in the sense of J.-\'E.~Pin~\cite{pin1992a}. In~other words, $\revlc{\mathbb{F}_2}$ is the \emph{variety generated by $\revlc{\mathbb{B}}$}. We thus see that introducing weights from~$\mathbb{F}_2$ to reversible finite automata
leads to describing a~larger class of languages with better closure properties. 

In fact, the Boolean closure of the positive variety of reversible languages $\revlc{\mathbb{B}}$ is known to correspond, via Eilenberg's correspondence, to the pseudovariety $\mathbf{ECom} = \mathbf{J_1} * \mathbf{G}$ of all finite monoids with commuting idempotents~\mbox{\cite{golovkins2006a,golovkins2010a,margolis1987a,ash1987a,pin1984a}}.
We thus obtain the following corollary, which can be seen both as a characterisation of the reversible languages over $\mathbb{F}_2$, as well as a new automata-theoretic characterisation of the pseudovariety $\mathbf{ECom}$. 

\begin{corollary}
Let $\Sigma$ be an alphabet a $L \subseteq \Sigma^*$ a language. Then $L$ belongs to $\revlc{\mathbb{F}_2}$ if and only if $ef = fe$ for any two idempotents $e,f$ in the syntactic monoid $M_L$ of $L$ over $\Sigma^*$.
Thus $L \in \revlc{\mathbb{F}_2}$ if and only if $M_L \in \mathbf{ECom}$. 
\end{corollary}

The pseudovariety $\mathbf{ECom}$ can also be described using the pseudoidentity $x^{\omega} y^{\omega} = y^{\omega} x^{\omega}$, which directly captures the property of commuting idempotents. In~any~case, the membership of a finite
monoid to $\mathbf{ECom}$ is clearly decidable, which implies decidability of the membership of a rational language to $\revlc{\mathbb{F}_2}$. 

\section{Automata over Locally Finite Commutative Rings}
                                                                                                         
In what follows, we show that the results just obtained for reversible weighted automata over $\mathbb{F}_2$ can actually be generalised to weighted automata over any nontrivial locally finite commutative ring.

We prove $\revlc{R} = \revlc{\mathbb{F}_2}$ for any such ring $R$, and in order to~establish one of the inclusions between these two classes, we show that any characteristic series of a language
from $\revlc{\mathbb{F}_2}$ over $R$ can be realised by a reversible weighted automaton over $R$. In fact, local finiteness and commutativity of~$R$ are~not necessary here. 

\begin{lemma}
\label{le:okruhy1}
Let $R$ be a nontrivial ring and $\Sigma$ an alphabet. Then for every $L \in \revl{\mathbb{F}_2}{\Sigma}$, the characteristic series $\underline{L}$ of $L$ over $R$ is in $\revs{R}{\Sigma}$, and~$L$ is in~$\revl{R}{\Sigma}$.  
\end{lemma} 
\begin{proof}
Let $L \in \revl{\mathbb{F}_2}{\Sigma}$; it suffices to show that $\underline{L} \in \revs{R}{\Sigma}$ over~$R$. As~$L \in \revl{\mathbb{F}_2}{\Sigma}$, it follows by~Proposition~\ref{prop:revs-char} and~by~Proposition~\ref{prop:chs2}
that there exists some $n \in \mathbb{N}$ and languages $L_1,\ldots,L_n \in \orevl{\mathbb{B}}{\Sigma}$ such that \mbox{$L = \supp(\underline{L_1} + \ldots + \underline{L_n})$} over $\mathbb{F}_2$.
This means that $w \in \Sigma^*$ is in $L$ if~and~only~if the~set $X_w = \{i \in [n] \mid w \in L_i\}$ contains an odd number of elements. Thus the~coefficients of the characteristic series $\underline{L}$ of $L$ over $R$ at $w \in \Sigma^*$ are given by
\begin{equation}
\label{eq:1}
(\underline{L}, w) = \left\{\begin{array}{ll}1 & \text{ if $\lvert X_w\rvert$ is odd,} \\ 0 & \text{ if $\lvert X_w\rvert$ is even.} \end{array}\right.
\end{equation}

Let us now consider the series $r \in R\llangle\Sigma^*\rrangle$ defined by
\begin{displaymath}
r = \sum_{\emptyset \subsetneq X \subseteq [n]} (-2)^{\lvert X\rvert - 1} \underline{\bigcap_{i \in X} L_i}. 
\end{displaymath}
Then $r \in \revs{R}{\Sigma}$ by virtue of Proposition~\ref{prop:ahk}, Proposition~\ref{prop:char-ser}, Proposition~\ref{prop:revs-char}, and~Proposition~\ref{prop:closure-sum}.
In order to prove that $\underline{L} \in \revs{R}{\Sigma}$, we now show that actually $r = \underline{L}$. Indeed, for $\emptyset \subsetneq X \subseteq [n]$ fixed and any $w \in \Sigma^*$, we have
$\left(\underline{\bigcap_{i \in X} L_i},w\right) = 1$ if $X \subseteq X_w$ and $\left(\underline{\bigcap_{i \in X} L_i},w\right) = 0$ otherwise,
while there are exactly ${\lvert X_w\rvert \choose k}$ sets $X \subseteq X_w$ of each size $k = 1,\ldots,\lvert X_w\rvert$. As~a~result,
\begin{displaymath}
(r, w) = \sum_{k = 1}^{\lvert X_w\rvert}{\lvert X_w\rvert \choose k} (-2)^{k - 1}
\end{displaymath} 
for each $w \in \Sigma^*$ over $R$. However, over $\mathbb{Q}$ it holds that
\begin{displaymath}
\sum_{k = 1}^{\lvert X_w\rvert}{\lvert X_w\rvert \choose k} (-2)^{k - 1} = -\frac{1}{2}\left(\sum_{k = 0}^{\lvert X_w\rvert}{\lvert X_w\rvert \choose k} (-2)^{k} - 1\right) = \frac{1 - (-1)^{\lvert X_w\rvert}}{2}
\end{displaymath}
by the Binomial theorem, which means that also over $R$ we have 
\begin{displaymath}
(r, w) = \sum_{k = 1}^{\lvert X_w\rvert}{\lvert X_w\rvert \choose k} (-2)^{k - 1} = \left\{\begin{array}{ll}1 & \text{ if $\lvert X_w\rvert$ is odd,} \\ 0 & \text{ if $\lvert X_w\rvert$ is even.} \end{array}\right.
\end{displaymath}
Thus $r = \underline{L}$ by (\ref{eq:1}).\qed\goodbreak
\end{proof}

We now essentially establish the remaining inclusion in case the nontrivial ring $R$ is \emph{locally finite} and \emph{commutative}. 

\begin{lemma}
\label{le:okruhy2}
Let $R$ be a nontrivial locally finite commutative ring, $\Sigma$ an~alphabet, \mbox{$L \in \revl{R}{\Sigma}$} a language, and $M_L$ the syntactic monoid of $L$ over $\Sigma^*$. 
Then the idempotents of $M_L$ commute, i.e., $M_L \in \mathbf{ECom}$.
\end{lemma}
\begin{proof}
As $\mathbf{ECom}$ is a pseudovariety of finite monoids and $M_L$ divides the transition monoid of any deterministic finite automaton recognising $L$, it suffices to~describe 
\emph{some} deterministic finite automaton $\mathcal{D}$ over $\Sigma$ recognising $L$ such that the idempotents in the transition monoid of $\mathcal{D}$ commute.

Since $L \in \revl{R}{\Sigma}$, there has to be some reversible weighted automaton~$\mathcal{A}$ over $R$ and $\Sigma$ with state set $[n]$ for some $n \in \mathbb{N}$ 
such that $L = \supp(\|\mathcal{A}\|)$; let~$\mathcal{P}_{\mathcal{A}} = (n,\mathbf{i},\mu,\mathbf{f})$ be the corresponding linear representation. For $k = 1,\ldots,n$, let $\mathbf{e}_k$ denote the vector $\mathbf{e}_k = (a_1,\ldots,a_n) \in R^n$ such that $a_k = 1$ and $a_j = 0$ for all $j \in [n] \setminus \{k\}$. 
As $\mathcal{A}$ is always an automaton over some finitely generated subring of $R$, we may assume that $R$ is actually finite.

We may thus take $\mathcal{D} = (R^n,\cdot,\mathbf{i},F)$, where $F = \{\mathbf{v} \in R^n \mid \mathbf{v}\mathbf{f} \neq 0\}$, and~\mbox{$\mathbf{v} \cdot a = \mathbf{v} \mu(a)$} for all $\mathbf{v} \in R^n$ and $a \in \Sigma$.
The automaton $\mathcal{D}$ then clearly recognises the language $L = \supp(\|\mathcal{A}\|)$ and its transition monoid can be represented as~a~monoid of all matrices $\mu(w)$ for $w \in \Sigma^*$ with matrix multiplication over $R$.
By reversibility of $\mathcal{A}$, each of the matrices $\mu(a)$ for $a \in \Sigma$ contains at~most one nonzero element in each row and column. The same property thus holds for~all matrices $\mu(w)$ for $w \in \Sigma^*$.
This means that if $\mu(w) = (a_{i,j})_{n \times n}$ is an~idempotent, then $\mu(w)$ is a diagonal matrix, as $a_{i,j} \neq 0$ for some $i \neq j$ would imply $\mathbf{e}_i \mu(w) = a_{i,j} \mathbf{e}_j$ and
$\mathbf{e}_i \mu(w)^2 = a_{i,j} \mathbf{e}_j \mu(w) \neq a_{i,j} \mathbf{e}_j$, as otherwise $a_{j,j}$ would have to be nonzero and the $j$-th column of $\mu(w)$ would contain two nonzero elements.
This would mean $\mu(w)^2 \neq \mu(w)$ and $\mu(w)$ would not be idempotent. Any idempotent $\mu(w)$ is thus indeed a diagonal matrix, while diagonal matrices over commutative rings commute under matrix multiplication.\qed   
\end{proof}

\begin{theorem}
\label{th:okruhy-hlavna}
The class $\revlc{R}$ is the Boolean closure of $\revlc{\mathbb{B}}$ for every nontrivial locally finite commutative ring $R$.
\end{theorem}
\begin{proof}
Follows by Lemma \ref{le:okruhy1}, Lemma \ref{le:okruhy2}, and Theorem \ref{th:f2hlavna} coupled with the~correspondence of the Boolean closure of $\revlc{\mathbb{B}}$ to $\mathbf{ECom}$.\qed 
\end{proof}

\begin{proposition}
Let $R$ be a nontrivial finite commutative ring, $\Sigma$ an alphabet, and~$r \in R\llangle\Sigma^*\rrangle$ a~series rational over $R$. Then $r \in \revs{R}{\Sigma}$ if and only if $\supp(r + x \cdot \underline{\Sigma^*}) \in \revl{R}{\Sigma}$ for all $x \in R$.  
\end{proposition}
\begin{proof}
If $r \in \revs{R}{\Sigma}$, then $r + x \cdot \underline{\Sigma^*} \in \revs{R}{\Sigma}$ for all $x \in R$ as well, hence $\supp(r + x \cdot \underline{\Sigma^*}) \in \revl{R}{\Sigma}$. If on the other hand $\supp(r + x \cdot \underline{\Sigma^*}) \in \revl{R}{\Sigma}$ for all $x \in R$,
we see that \mbox{$r = \sum_{x \in R} x\cdot\underline{\left(\Sigma^* \setminus \supp(r - x \cdot \underline{\Sigma^*})\right)} \in \revs{R}{\Sigma}$}
by~\mbox{Theorem}~\ref{th:okruhy-hlavna}, Lemma \ref{le:okruhy1}, and Proposition \ref{prop:closure-sum} (observe that we have just expressed $r$ as a specific recognisable step function \cite{droste2007a,droste2009c}).\qed
\end{proof}

\begin{corollary}
The problem of checking reversibility of a rational series over an~effective locally finite commutative ring $R$ is decidable.
\end{corollary}\goodbreak


\bibliographystyle{splncs04}
\bibliography{references}

\begin{thebibliography}{10}
\providecommand{\url}[1]{\texttt{#1}}
\providecommand{\urlprefix}{URL }
\providecommand{\doi}[1]{https://doi.org/#1}

\bibitem{almeida1994a}
Almeida, J.: Finite Semigroups and Universal Algebra. World Scientific (1994)

\bibitem{ambainis1998a}
Ambainis, A., Freivalds, R.M.: 1-way quantum finite automata: Strengths,
  weaknesses and generalizations. In: Foundations of Computer Science, FOCS
  1998. pp.~332--341 (1998)

\bibitem{angluin1982a}
Angluin, D.: Inference of reversible languages. Journal of the {ACM}
  \textbf{29}(3),  741--765 (1982)

\bibitem{ash1987a}
Ash, C.J.: Finite semigroups with commuting idempotents. Journal of the
  Australian Mathematical Society (Series A)  \textbf{43}(1),  81--90 (1987)

\bibitem{axelsen2016a}
Axelsen, H.B., Holzer, M., Kutrib, M.: The degree of irreversibility in
  deterministic finite automata. In: Implementation and Application of
  Automata, CIAA 2016. pp.~15--26 (2016)

\bibitem{axelsen2017a}
Axelsen, H.B., Holzer, M., Kutrib, M.: The degree of irreversibility in
  deterministic finite automata. International Journal of Foundations of
  Computer Science  \textbf{28}(5),  503--522 (2017)

\bibitem{bakinova2022a}
Bakinova, E., Basharin, A., Batmanov, I., Lyubort, K., Okhotin, A., Sazhneva,
  E.: Formal languages over {GF(2)}. Information and Computation  \textbf{283}
  (2022), article~104672

\bibitem{bala2013a}
Bala, S.: Which finitely ambiguous automata recognize finitely sequential
  functions? In: Mathematical Foundations of Computer Science, MFCS 2013. pp.
  86--97 (2013)

\bibitem{bell2023a}
Bell, J.P., Smertnig, D.: Computing the linear hull: Deciding {D}eterministic?
  and {U}nambiguous? for weighted automata over fields. In: Logic in Computer
  Science, LICS 2023 (2023)

\bibitem{benalioua2024a}
Benalioua, Y.I., Lhote, N., Reynier, P.A.: Minimizing cost register automata
  over a field. In: Mathematical Foundations of Computer Science, MFCS 2024
  (2024), article~23

\bibitem{bennett1973a}
Bennett, C.H.: Logical reversibility of computation. {IBM} Journal of Research
  and Development  \textbf{17}(6),  525--532 (1973)

\bibitem{berstel2011a}
Berstel, J., Reutenauer, C.: Noncommutative Rational Series with Applications.
  Cambridge University Press (2011)

\bibitem{droste2007a}
Droste, M., Gastin, P.: Weighted automata and weighted logics. Theoretical
  Computer Science  \textbf{380}(1--2),  69--86 (2007)

\bibitem{droste2009c}
Droste, M., Gastin, P.: Weighted automata and weighted logics. In: Droste, M.,
  Kuich, W., Vogler, H. (eds.) Handbook of Weighted Automata, chap.~5, pp.~\mbox{175--211}. Springer (2009)

\bibitem{droste2009a}
Droste, M., Kuich, W., Vogler, H. (eds.): Handbook of Weighted Automata.
  Springer (2009)

\bibitem{droste2021a}
Droste, M., Kuske, D.: Weighted automata. In: Pin, J.{\'E}. (ed.) Handbook of
  Automata Theory, Vol. 1, chap.~4, pp. 113--150. European Mathematical Society
  (2021)

\bibitem{garcia2009a}
Garc{\'i}a, P., de~Parga, M.V., Cano, A., L{\'o}pez, D.: On locally reversible
  languages. Theoretical Computer Science  \textbf{410},  4961--4974 (2009)

\bibitem{golovkins2006a}
Golovkins, M., Pin, J.{\'E}.: Varieties generated by certain models of
  reversible finite automata. In: Computing and Combinatorics, COCOON 2006. pp.
  83--93 (2006)

\bibitem{golovkins2010a}
Golovkins, M., Pin, J.{\'E}.: Varieties generated by certain models of
  reversible finite automata. Chicago Journal of Theoretical Computer Science
  \textbf{2010} (2010), article~2

\bibitem{guillon2022a}
Guillon, B., Lavado, G.J., Pighizzini, G., Prigioniero, L.: Weakly and strongly
  irreversible regular languages. International Journal of Foundations of
  Computer Science  \textbf{33}(3--4),  263--284 (2022)

\bibitem{holzer2015a}
Holzer, M., Jakobi, S., Kutrib, M.: Minimal reversible deterministic finite
  automata. In: Developments in Language Theory, DLT 2015. pp. 276--287 (2015)

\bibitem{holzer2018a}
Holzer, M., Jakobi, S., Kutrib, M.: Minimal reversible deterministic finite
  automata. International Journal of Foundations of Computer Science
  \textbf{29}(2),  251--270 (2018)

\bibitem{holzer2017a}
Holzer, M., Kutrib, M.: Reversible nondeterministic finite automata. In:
  Reversible Computation, RC 2017. pp. 35--51 (2017)

\bibitem{jecker2024a}
Jecker, I., Mazowiecki, F., Purser, D.: Determinisation and unambiguisation of~polynomially-ambiguous rational weighted automata. In: Logic in Computer
  Science, LICS 2024 (2024)

\bibitem{kiefer2021a}
Kiefer, S., Widdershoven, C.: Image-binary automata. In: Descriptional
  Complexity of Formal Systems, DCFS 2021. pp. 176--187 (2021)

\bibitem{kiefer2024a}
Kiefer, S., Widdershoven, C.: Image-binary automata. International Journal of~Foundations of Computer Science  (2024), to appear

\bibitem{klima2018a}
Kl{\'i}ma, O., Pol{\'a}k, L.: Forbidden patterns for ordered automata. In:
  Non-Classical Models of Automata and Applications, NCMA 2018. pp. 99--115
  (2018)

\bibitem{klima2020b}
Kl{\'i}ma, O., Pol{\'a}k, L.: Forbidden patterns for ordered automata. Journal
  of Automata, Languages and Combinatorics  \textbf{25}(2--3),  141--169 (2020)

\bibitem{kostolanyi2022b}
Kostol{\'a}nyi, P.: Bideterministic weighted automata. In: Algebraic
  Informatics, CAI~2022. pp. 161--174 (2022)

\bibitem{kostolanyi2022a}
Kostol{\'a}nyi, P.: Determinisability of unary weighted automata over the
  rational numbers. Theoretical Computer Science  \textbf{898},  110--131
  (2022)

\bibitem{kostolanyi2022c}
Kostol{\'a}nyi, P.: Finite ambiguity and finite sequentiality in weighted
  automata over fields. In: Computer Science -- Theory and Applications, CSR
  2022. pp. 209--223 (2022)

\bibitem{kostolanyi2023c}
Kostol{\'a}nyi, P.: Bideterministic weighted automata. Information and
  Computation  \textbf{295B} (2023), article~105093

\bibitem{kostolanyi2024a}
Kostol{\'a}nyi, P.: Finitely ambiguous and finitely sequential weighted
  automata over fields. Theoretical Computer Science  \textbf{1012} (2024),
  article~114725

\bibitem{landauer1961a}
Landauer, R.: Irreversibility and heat generation in the computing process.
  {IBM} Journal of Research and Development  \textbf{5}(3),  183--191 (1961)

\bibitem{lombardy2002a}
Lombardy, S.: On the construction of reversible automata for reversible
  languages. In: Automata, Languages and Programming, ICALP 2002. pp. 170--182
  (2002)

\bibitem{lombardy2006a}
Lombardy, S., Sakarovitch, J.: Sequential? Theoretical Computer Science
  \textbf{356},  \mbox{224--244} (2006)

\bibitem{margolis1987a}
Margolis, S.W., Pin, J.{\'E}.: Inverse semigroups and varieties of finite
  semigroups. Journal of Algebra  \textbf{110}(2),  306--323 (1987)

\bibitem{okhotin2024a}
Okhotin, A., Radionova, M., Sazhneva, E.: {GF(2)}-operations on basic families
  of~formal languages. Theoretical Computer Science  \textbf{995} (2024),
  article~114489

\bibitem{paul2019a}
Paul, E.: Finite sequentiality of unambiguous max-plus tree automata. In:
  Symposium on Theoretical Aspects of Computer Science, STACS 2019. pp.
  55:1--55:17 (2019)

\bibitem{paul2020a}
Paul, E.: Finite sequentiality of finitely ambiguous max-plus tree automata.
  In:~Automata, Languages and Programming, ICALP 2020. pp. 137:1--137:15 (2020)

\bibitem{paul2021a}
Paul, E.: Finite sequentiality of unambiguous max-plus tree automata. Theory of~Computing Systems  \textbf{65}(4),  736--776 (2021)

\bibitem{pin1984a}
Pin, J.{\'E}.: Hi{\'e}rarchies de concat{\'e}nation. {RAIRO}. Informatique
  th{\'e}orique.  \textbf{18}(1),  23--46 (1984)

\bibitem{pin1992a}
Pin, J.{\'E}.: On reversible automata. In: Latin American Symposium on
  Theoretical Informatics, LATIN 1992. pp. 401--416 (1992)

\bibitem{pin1997a}
Pin, J.{\'E}.: Syntactic semigroups. In: Rozenberg, G., Salomaa, A. (eds.)
  Handbook of Formal Languages, vol.~1, chap.~10, pp. 679--746. Springer (1997)

\bibitem{radionova2024a}
Radionova, M., Okhotin, A.: Decision problems for reversible and permutation
  automata. In: Implementation and Application of Automata, CIAA~2024. pp.~\mbox{302--315} (2024)

\bibitem{radionova2025a}
Radionova, M., Okhotin, A.: A hierarchy of reversible finite automata. In:
  Implementation and Application of Automata, CIAA 2025. pp. 316--329 (2025)

\bibitem{sakarovitch2009a}
Sakarovitch, J.: Elements of Automata Theory. Cambridge University Press (2009)

\bibitem{salomaa1978a}
Salomaa, A., Soittola, M.: Automata-Theoretic Aspects of Formal Power Series.
  Springer (1978)

\bibitem{schutzenberger1961a}
Sch{\"u}tzenberger, M.P.: On the definition of a family of automata.
  Information and~Control  \textbf{4}(2--3),  245--270 (1961)

\bibitem{straubing2021a}
Straubing, H., Weil, P.: Varieties. In: Pin, J.{\'E}. (ed.) Handbook of
  Automata Theory, Vol. 1, chap.~16, pp. 569--614. European Mathematical
  Society (2021)

\bibitem{tamm2008a}
Tamm, H.: On transition minimality of bideterministic automata. International
  Journal of Foundations of Computer Science  \textbf{19}(3),  677--690 (2008)

\bibitem{tamm2003a}
Tamm, H., Ukkonen, E.: Bideterministic automata and minimal representations of~regular languages. In: Implementation and Application of Automata, CIAA~2003.
  pp. 61--71 (2003)

\bibitem{tamm2004a}
Tamm, H., Ukkonen, E.: Bideterministic automata and minimal representations of~regular languages. Theoretical Computer Science  \textbf{328}(1--2),
  135--149 (2004)

\bibitem{merwe2012a}
{van der Merwe}, B., Tamm, H., {van Zijl}, L.: Minimal {DFA} for symmetric
  difference {NFA}. In: Descriptional Complexity of Formal Systems, DCFS 2012.
  pp. 307--318 (2012)

\bibitem{zijl2004a}
{van Zijl}, L.: On binary $\oplus$-{NFA}s and succinct descriptions of regular
  languages. Theoretical Computer Science  \textbf{328}(1--2),  161--170 (2004)

\end{thebibliography}

\end{document}